\documentclass[]{amsart}
\usepackage{amsmath, amsthm, amscd, amsfonts,amssymb,amscd, parskip, xcolor, hyperref}
\newtheorem{theorem}{Theorem}[section]
\newtheorem{proposition}[theorem]{Proposition}

\newtheorem{fact}[theorem]{Fact}

\newtheorem{lemma}[theorem]{Lemma}
\theoremstyle{definition}
\newtheorem{definition}[theorem]{Definition}
\newtheorem{example}[theorem]{Example}

\newtheorem{remark}[theorem]{Remark}

\numberwithin{equation}{section}
\begin{document}
\title[Cryptanalysis and Improvement of Akleylek et al.'s Cryptosystem]
{Cryptanalysis and Improvement of Akleylek et al.'s cryptosystem}

\author[R. Rastaghi]{R. Rastaghi}
\address{Department of Electrical Engineering, Aeronautical University of Sciences \& Technology, Tehran, Iran.}
\email{r.rastaghi59@gmail.com}
\begin{abstract}
Akleylek et al. [S. Akleylek, L. Emmungil and U. Nuriyev, A modified algorithm for peer-to-peer security, {\it journal of Appl. Comput. Math.}, vol. 6(2), pp.258-264, 2007.], introduced a modified public-key encryption scheme with steganographic approach for security in peer-to-peer (P2P) networks. In this cryptosystem, Akleylek et al. attempt to increase security of the P2P networks by mixing ElGamal cryptosystem with knapsack problem. In this paper, we present a ciphertext-only attack against their system to recover message. In addition, we show that
for their scheme \textit{completeness} property is not holds, and therefore, the receiver cannot \textit{uniquely} decrypts messages. Furthermore, we also show that this system is not chosen-ciphertext secure, thus the proposed scheme is vulnerable to man-in-the-middle-attack, one of the most pernicious attacks against
P2P networks. Therefore, this scheme is not suitable to implement in the P2P networks.\\
We modify this cryptosystem in order to increase its security and efficiency. Our construction is the
efficient CCA2-secure variant of the Akleylek et al.'s encryption scheme in the standard model, the \textit{de facto} security notion for public-key encryption schemes.
\end{abstract}
\maketitle
\section{Introduction}
The use of computer network is raised day by day. This increment causes the number of nodes to increase. By increasing the client, the server becomes busy and insufficient although
the bandwidths are high enough. Moreover, since the variety of requests is increased, servers may not have data the user needs. We can overcome these obstacles by using peer-to-peer
(P2P) network. The P2P networks have become popular as a new paradigm for information exchange and are being used in many applications such as file sharing, distributed computing, video
conference, VoIP, radio and TV broadcasting. The P2P networks did not have centralized servers; some powerful nodes act as server. The fourth generation supports streams over P2P networks and each node can talk with another. In these networks, since server has been decentralized and each node can directly communicate with other nodes, management and security become a most important problem. There are several ways to make P2P networks secure. Cryptography plays the most important role in each way. Cryptography is the art of keeping the data secure from eavesdropping and other malicious activities. Therefore, cryptographic algorithms are very essential in the P2P systems since they can uniquely protect message for an individual, and verify its integrity.\\
Due to peer-relying nature of the P2P networks, they are susceptible to many general attacks. Man-in-the-middle attack is one
the most pernicious attacks against P2P networks. The man-in-the-middle attack is an indirect intrusion, where the attacker
inserts its node undetected between two nodes. It is typically used for eavesdropping a public-key encrypted conversation to retrieve, modify or cut the information being sent by adopting some strategies and tricks. Therefore, the public-key encryption (PKE) scheme must resist against this type of powerful attack. Security against adaptive chosen-ciphertext attack (i.e.,
CCA2 security) \cite{17} is the strong security notion for a PKE scheme. This notion is known to suffice for many applications of encryption in the presence of active attackers
--- a man-in-the-middle adversary --- including: secure P2P transmission, secure communication, auctions, voting schemes,
and many others. In this scenario, the adversary has seen \textit{challenge ciphertext} before having access to the decryption oracle. The adversary is not allowed to ask the decryption of
the challenge ciphertext, but can obtain the decryption of any relevant ciphertext {\it even modified ones based on the challenge ciphertext}. A cryptosystem is CCA2-secure
if the cryptanalyst fails to obtain any partial information about the plaintext relevant to the challenge ciphertext. The most cryptographic protocols cannot prevent chosen-ciphertext
attacks mounted by a man-in-the-middle adversary who has full control of the communication channel between the sender and the receiver. Indeed, design \textit{efficient} CCA2-secure
encryption scheme is a challenging problem in cryptography.

In \cite{2}, Akleylek et al. introduced‎ a modified algorithm with steganographic approach for security in the P2P networks. In this cryptosystem, Akleylek et al. attempt to increase security of the P2P system by mixing ElGamal cryptosystem~\cite{8} with knapsack problem. The knapsack problem is a decision problem which is NP-complete~\cite{11,12,13}. That is to say, this problem cannot be easily solved even using quantum computers. They use the ElGamal encryption scheme to disguise private knapsack (easy knapsack) in order to produce public key (hard knapsack). In this paper, we show that this combination leaks the security and makes the cryptosystem vulnerable to ciphertext-only attack. Any encryption scheme vulnerable to this type of attacks is considered to be completely insecure. In addition, we show that in most cases \textit{completeness} property does not holds for their system. Therefore, the receiver cannot {\it uniquely} decrypts ciphertexts. Besides, their construction is deterministic and so each message has one primage. Therefore, an attacker can simply distinguish between decryptions of the two different messages. Hereupon, this encryption scheme does not satisfies indistinguishability (a.k.a semantic security) against chosen ciphertext attack\footnote{Randomized encryption algorithm is a necessary condition for CCA2 security. Although randomness is necessary, it is not sufficient (see subsection \ref{ssec2.4}).}. Hence, in the network an attacker can apply these attacks and simply can recover plaintext from any challenge ciphertext. Thereupon, this scheme is not suitable for using in a P2P network. We propose a modification to this scheme in order to increase security, efficiency and usability for using in the P2P networks. Our construction is a CCA2-secure PKE scheme in the standard model, the \textit{de facto} security notion for PKE schemes. The main novelty is that scheme's {\it consistency} check can be directly implemented by the system
without having access to some external gap-oracle as in \cite{3,4} or using other extrinsic rejection techniques \cite{6}.
\subsection{Related works}
In 1998, Cai and Cusick \cite{5} proposed an efficient lattice-based public-key cryptosystem
with much less data expansion by mixing the Ajtai-Dwork cryptosystem \cite{1} with an additive knapsack. Recently, their cryptosystem was broken by Pan and Deng \cite{16}. They presented an iterative
method to recover the message encrypted by the Cai-Cusick cryptosystem under a ciphertext-only scenario. They also present two chosen-ciphertext attacks to get a similar private key
which acts as the real private key. In another work, with several known attacks in mind, very recently Pan et al. \cite{15} introduced a new lattice-based PKE scheme mixed with additive knapsack
problem which has reasonable key size and quick encryption and decryption. Unfortunately, their scheme was broken by Xu et al. \cite{19}. They proposed two feasible attacks on the
cryptosystem of Pan et al.; the first one is a broadcast attack assuming a single encrypted message directed towards for several recipients with different public keys,
the message can be recovered by solving a system of nonlinear equations via linearization technique. The second one is a multiple transmission attack in which a single message is
encrypted under the same public key for several times using different random vectors. In this situation, the message can be easier to recover. Very recently, Rasatghi \cite{18}
introduced an efficient PKE scheme which is robust against man-in-the-middle adversaries for the P2P networks. His scheme uses RSA cryptosystem in combination of the additive knapsack
problem. Since RSA encryption scheme is deterministic and therefore does not satisfies CCA2 security requirements, the encryption algorithm uses a new padding scheme for encoding
input messages in order to secure mixed scheme against chosen-ciphertext attack.

\textbf{Organization.} The rest of this paper is organized as follows: In the next section, we give some mathematical background and definitions. Akleylek et al.'s cryptosystem will be presented in section 3. Section 4 presents our cryptanalysis and in section 5, we modify this cryptosystem to achieve desired security i.e., CCA2-security and efficiency. Some conclusion is given in section 6.
\section{Preliminaries}
\subsection{Notation}
We will use standard notation. If $x$ is a string, then $\left| x \right|$ denotes its length. If $k \in \mathbb{N}$, then $\left\{ {0,\,1} \right\}^k$ denote the set of \textit{k}-bit strings,
$1^k$ denote a string of $k$ ones and $\left\{ {0,\,1} \right\}^*$ denote the set of bit strings of finite length. $y \leftarrow x$ denotes the assignment to \textit{y }of the value
\textit{x}. For a set $S$, $s \leftarrow S$ denote the assignment to $s$ of a uniformly random element of $S$. For a deterministic algorithm ${\mathcal A}$, we write
$x\leftarrow {\mathcal A}^{\mathcal O} (y,\,z)$ to mean that \textit{x} is assigned the output of running ${\mathcal A}$ on inputs \textit{y} and \textit{z}, with access to oracle ${\mathcal O}$.
We denote by ${\rm Pr}[E]$ the probability that the event $E$ occurs.
\subsection{Mathematical background}
\begin{definition} [\textbf{Subset sum problem} \footnote{Additive knapsack problem.}] Given a set of positive integers $(a_1, \ldots, a_n)$ and a positive integer $s$. Whether there is a subset of the $a_i$s such that their sums equal to $s$. That is equivalent to determine whether there are variables $(x_1, \ldots, x_n)$ such that
$$s=\sum_{i=1}^{n}a_ix_i, \quad x_i\in \{0, 1\}, \quad 1\leq i\leq n. $$
\end{definition}
The subset sum ($0-1$ knapsack) is a decision problem which is NP-complete. The computational version of the subset sum problem is NP-hard \cite{13}
\begin{definition}[\textbf{Super-increasing sequence}] The sequence $(a_1, \ldots, a_n)$ of positive integers is a super increasing sequence if $a_i > \sum_{j=1}^{i-1}a_j$ for all $i\geq 2$.
\end{definition}
There is an efficient greedy algorithm to solve the subset sum problem if the $b_i$s are a super-increasing sequence: Just subtract the largest possible value from $s$ and repeat. The following algorithm efficiently solves the subset sum problem for super-increasing sequences in the polynomial time.

\begin{quote}
    {\bf Algorithm 1} Solving a super-increasing subset sum problem.
\end{quote}
\begin{quote}
\begin{quote}
    {\it Input:} Super-increasing sequence $(a_1, \ldots, a_n)$ and an integer $s$ which is the sum of a subset of the $a_i$. \\
    {\it Output:} $(x_1, \dots, x_n)$ where $x_i\in \{0, 1\}$, such that $s=\sum_{i=1}^{n}a_ix_i$.
\begin{enumerate}
             \item $i\leftarrow n$
             \item While $i\geq 1$ do the following:
             \begin{enumerate}
               \item If $s\geq a_i$, then $x_i \leftarrow 1$ and $s\leftarrow s-a_i$. Otherwise $x_i\leftarrow 0$.
               \item $i\leftarrow i-1.$
             \end{enumerate}
             \item Return $(x_1, \ldots, x_n)$.
           \end{enumerate}
\end{quote}
\end{quote}
\begin{definition}[\bf {Subset product problem}\footnote{Multiplicative knapsack problem.}] A set of positive integers $(a_1, \dots, a_n)$ and a positive integer $d$ are given. Whether there is a subset of the $a_i$'s such that their product equals to $d$. That is equivalent to determine whether there are variables $(x_1, \dots, x_n)$ such that
$$d=\prod_{i=1}^{n}a_i^{x_i}, \quad x_i\in \{0, 1\}, \quad 1\leq i\leq n. $$
\end{definition}
The multiplicative knapsack (subset product) problem is a decision problem which is NP-complete \cite{11,12}. As observed in \cite{10,11,12,14}, if the $a_i$s are relatively prime, then this problem can be solved in polynomial time by factoring $d$. Their result can be summarized in the following lemma.
\begin{lemma}\label{l1}
If $(a_1, a_2, \dots, a_n)$ are relatively prime, then we can solve subset product problem in the polynomial time.
\end{lemma}
\begin{proof}
Since the $a_i$s are relatively prime and $x_i\in \{0, 1\}$, so we have
\[
x_i=\left\{ \begin{array}{ll}
 1 \quad {\rm if} \quad {\sf gcd}(d, a_i)=a_i \\
 0\quad {\rm if} \quad {\sf gcd}(d, a_i)=1
  \end{array} \right., \quad 1\leq i\leq n
\]
Hence,
\[
x_i=\left\{ \begin{array}{ll}
 1 \quad {\rm if} \quad a_i \mid d\\
 0 \quad {\rm if} \quad a_i \nmid d
  \end{array} \right., \quad 1\leq i\leq n
\]
where {\sf gcd} means the greatest common divisor.
\end{proof}
\begin{definition}[\textbf{Discrete logarithm problem (DLP)}]
Given a prime ‎$p$‎, a generator‎ $g$ of $\mathbb{Z}‎_p^*$ ‎, and an element‎ $y \in \mathbb{Z}‎_p^*$. Find integer $x, 0\leq x\leq p-2$‎‎, such that
\[‎
 g ^x=y \mod p‎.‎
\]
is called the discrete logarithm problem.
\end{definition}
\begin{fact}\label{fa2.6}
Suppose that $g$ is a generator of $\mathbb{Z}_p^*$. Then $b = g^i \mod p$ is also a generator of $\mathbb{Z}_p^*$ if and only if $\gcd(i, p-1)) = 1$.
\end{fact}
\begin{definition}
A safe prime $p$ is a prime of the form $p = 2q + 1$ where $q$ is also prime.
\end{definition}
\subsection{Definitions}\label{de2.7}
\begin{definition} [{\bf Public-key encryption scheme}] A PKE scheme is a triple of probabilistic polynomial time
 (PPT) algorithms ({\sf Gen},\,{\sf Enc},\,{\sf Dec}) such that:
\begin{itemize}
\item {\sf Gen} is a probabilistic polynomial-time key generation algorithm which takes a security parameter $1^n$ as input and
outputs a public key $pk$ and a secret key $sk$. We write $(pk,sk) \leftarrow {\sf Gen}(1^n )$. The public key specifies the
message space ${\mathcal M}$ and the ciphertext space ${\mathcal C}$.
\item {\sf Enc} is a $($possibly$)$ probabilistic polynomial-time encryption algorithm which takes as input a public key $pk$, a
$m\in {\mathcal M}$ and random coins $r$, and outputs a ciphertext $C \in {\mathcal C}$. We write $C\leftarrow {\sf Enc}(pk,m; r)$ to
indicate explicitly that the random coins $r$ is used and $C \leftarrow {\sf Enc}(pk,m)$ if fresh random coins are used.
\item {\sf Dec} is a deterministic polynomial-time decryption algorithm which takes as input a secret-key $sk$ and a ciphertext
$C \in {\mathcal C}$, and outputs either a message $m \in {\mathcal M}$ or an error symbol $\bot$. We write $m \leftarrow{\sf Dec}(C,\,sk)$.
\item \textsf{Completeness}. For any pair of public and secret keys generated by {\sf Gen} and any message $m \in {\mathcal M}$ it holds that
${\sf Dec}(sk,\,{\sf Enc}(pk,m;r))=m$ with overwhelming probability over the randomness used by {\sf Gen} and the random coins \textit{r} used by {\sf Enc}.
\end{itemize}
\end{definition}
\begin{definition}[\textbf{Ciphertext-only attack}]
A ciphertext-only attack is a scenario by which the adversary ‎‎(or cryptanalyst) tries to deduce the decryption key by only observing the ciphertexts or decrypt a {\it challenge ciphertext}.‎
\begin{description}
  \item[Attacker knowledge] some $y_1= {\sf Enc}(x_1, pk)$, $y_2={\sf Enc}(x_2, pk)$, $\dots$.
  \item[Attacker goal] obtain $x_1, x_2, \dots$ or the secret-key $sk$.
\end{description}
Any encryption scheme vulnerable to this type of attacks is considered to be completely insecure.
\end{definition}
\begin{definition} [{\bf CCA2-security}]\label{de2d10}
A PKE scheme is secure against adaptive chosen-ciphertext attacks (i.e. IND-CCA2) if the advantage of any two-stage PPT
adversary ${\mathcal A} = ({\mathcal A}_1 ,\,{\mathcal A}_2 )$ in the following experiment is negligible in the security parameter $k$:
\begin{eqnarray*}
&&{\bf Exp}_{{\rm PKE}, {\mathcal A}}^{cca2}(k):\\
&& \quad\quad (pk, sk)\leftarrow {\sf Gen}(1^k)\\
&& \quad\quad (m_0, m_1, {\sf state})\leftarrow {\mathcal A}_{1}^{{\sf Dec}(sk, \cdot)} (pk) \quad {\rm s.t.} \quad |m_0|=|m_1|\\
&& \quad \quad b\leftarrow \{0, 1\}\\
&& \quad \quad C^{*}\leftarrow {\sf Enc}(pk, m_b)\\
&& \quad \quad b'\leftarrow {\mathcal A}_{2}^{{\sf Dec}(sk, \cdot)}(C^{*}, {\sf state})\\
&& \quad \quad  {\rm if} \ \ b= b^{'} \, \  {\rm return \ 1, \ else \ return \ 0}.
\end{eqnarray*}
The attacker may query a decryption oracle with a ciphertext $C$ at any point during its execution, with the exception that ${\mathcal A}_2 $ is not allowed to query
${\rm Dec} (sk,\,\cdot)$ with \textit{challenge ciphertext} $ C^{*} $. The decryption oracle returns $ b' \leftarrow {\mathcal A}_2^{{\rm Dec} (sk,\cdot)} (C^{*} ,{\sf state})$.
The attacker wins the game if $b = b'$ and the probability of this event is defined as $ \Pr [{\rm Exp} \,_{{\rm PKE} ,\,{\mathcal A}}^{cca2} \,(k)] $.
 We define the advantage of $ {\mathcal A}$ in the experiment as
\[
{\sf Adv}_{{\rm PKE} ,\,{\mathcal A}}^{\rm cca2} \,(k) = \left| \Pr [{\rm Exp} \,_{{\rm PKE} ,\,{\mathcal A}}^{cca2} \,(k) = 1] - \frac{1}{2} \right|.
\]
\end{definition}

\subsection{ElGamal Cryptosystem}\label{ssec2.4}
The ElGamal cryptosystem \cite{8} is a PKE scheme based on discrete logarithm problem (DLP) in $(\Bbb{Z}_p^{*}, \cdot)$. Let $p$ ba a large prime such that the DLP is infeasible in $(Z_p^{*}, \cdot)$, and let $g\in \Bbb{Z}_p^{*}$ be a primitive element. Each user selects a random integer $x$, $1\leq x\leq p-2$, and computes $y=g^x \mod p$. $(p, g, y)$ is the public key and $x$ is the secret key.\\
For encrypts a message, the sender randomly chooses integer $r, 1\leq r\leq p-2$ and computes $C_1=g^r,\, C_2=my^r$ and send $C=(C_1,C_2)$ to the receiver.
To recover message $m$ from ciphertext $C$, the receiver using private key $x$ computes $m=C_2(C_1^x)^{-1} \mod p$.

Altough the ElGamal scheme is randomized, but it not CCA2-secure. An attacker can pick a random number $r'$ and generate the ciphertext $C_1' = g^{r+r'}, \,C_2' = my^{r+r'}=mg^{x(r+r')}$, as the
values $g$ and $y$ are known from the public key. The attacker can then query for the decryption of this modified ciphertext and receive the message $m$ as answer.
\section{Akleylek et al. Cryptosystem}
In this section, we present Akleylek et al. cryptosystem \cite{2}. They wish to increase security of proposed cryptosystem by mixing the ElGamal cryptosystem with multiplicative knapsack problem.
\begin{enumerate}
\item \textbf{Key generation}
\begin{itemize}
  \item[(a)] We choose a super-increasing sequence $A=(a_1, \dots, a_n)$, such that $a_i > \sum_{i=1}^{j-1}a_i$, $2\leq j \leq n$, and all $a_i$'s are integer.
  \item[(b)] The keys of the ElGamal cryptosystem $(y, g, p, x)$ are calculated, where $y=g^x$.
  \item[(c)] For calculating public knapsack $B=(b_0, \dots, b_n)$, randomly select an integer $k$, $1\leq k\leq p-2$ and compute:
  \begin{eqnarray*}
  &&y=g^x \mod p, \quad s_i=g^k \mod p,\\
  && u_i=y^k.a_i \mod p, \quad b_i=(s_i, u_i) \quad {\rm for} \ \ \ 1\leq i\leq n.
\end{eqnarray*}
\end{itemize}
Finally, $B=(b_1, \dots, b_n)=((s_1, u_1), \dots, (s_n, u_n))$ is the public key and $(y, g, p, x, (a_1, \dots, a_n))$ is the secret key.
\item \textbf{Encryption}\\
To encrypt $n$ bit binary message $m=(m_1, \dots, m_n)$, we compute
\begin{equation}\label{e1}
C=(C_1, C_2)=\prod_{i=1}^{n}(s_i, u_i)^{m_i},
\end{equation}
and send ciphertext $C$ to the receiver.
\item \textbf{Decryption}\\
To decrypt the ciphertext $C$, the receiver firstly calculates
\begin{equation}\label{e2}
d=C_2.(C_1^x)^{-1} \mod p=\frac{\prod_{i=1}^n u_i^{m_i}}{\prod_{i=1}^n (s_i^x)^{m_i} }\mod p=\prod_{i=1}^n a_i^{m_i} \mod p.
\end{equation}
After calculating $d$, we must obtain plaintext $m=(m_1, \dots, m_n)$ from $d=a_1^{m_1}a_2^{m_2}\dots a_n^{m_n}$. Note that $u_i=y^k.a_i \mod p =g^{xk}.a_i \mod p=(s_i)^{x}.a_i \mod p$.
\end{enumerate}

\begin{remark}\label{re1}
\textit{ We stress that for the decryption algorithm works, we need to choose prime $p$ such that $p\geq \prod_{i=1}^n a_i$, which does not remark on the Akleylek et al.'s original paper. We illustrate this with an example in the next subsection.}
\end{remark}
\subsection{On the Completeness of the Akleylek et al.'s Cryptosystem}\label{ssec3.1}
The Akleylek et al.'s cryptosystem has some ambiguity. Completeness property for a PKE scheme (Definition \ref {de2.7}) guarantees that for any message $m \in {\mathcal M}$ it holds that ${\sf Dec}(sk,\,{\sf Enc}(pk,m))=m$. In the Akleylek et al.'s cryptosystem, after apply secret key we have $d=C_2.(C_1^x)^{-1} \mod p=\prod_{i=1}^n a_i^{m_i} \mod p$\,\footnote{As we mentioned in Remark \ref{re1}, we suppose that $p\geq \prod_{i=1}^n a_i$ and therefore $d=\prod_{i=1}^n a_i^{m_i} \mod p=\prod_{i=1}^n a_i^{m_i}$ and we have no problem for decrypting the input messages. See Example \ref{ex3.3} for more details.}, where $a_i, 1\leq i\leq n$ is a supper-increasing sequence. If the Hamming weight of the input message is small, for small $a_i$s we can efficiently retrieve the input messages but for large Hamming weight, $d$ is the product of the large subset of the $a_1, \dots, a_n$ and therefore it maybe impossible for the receiver to efficiently recovers $m_i$s from $d$. The main drawback is that the small $a_i$s maybe the divisors of the larger $a_i$s and therefore a ciphertext maybe does not decrypted {\it uniquely} and has several decryptions. A moment's reflection reveals that if we want any ciphertext decrypts uniquely, the $a_i$s must be \textit{pairwise primes}. Therefore, super-increasing assumption on the $a_i$s is not sufficient for completeness of the PKE scheme and their system does not satisfies completeness property. We illustrate our claims with a small example.

\begin{example} \label{ex3.2}
Suppose $(a_1,a_2,a_3,a_4,a_5)=(2,3,6, 12, 24)$ be a super-increasing sequence. Let $p = 2579$ and $g=2$,  where $g$ is a generator of $\mathbb{Z}_{2579}^{*}$. If we randomly choose $k=348$ and $x=1500$, then we have:
\begin{eqnarray*}
&&y=g^x \mod p=2^{1500} \mod 2579=862,\\
&& s_i=g^k \mod p=2^{348} \mod 2579=104,\\
&&u_1=y^k\times a_1\mod p =862^{348}\times 2 \mod 2579=2165,\\
&&u_2=1958, \quad u_3=1337, \quad u_4=95, \quad u_5=190.
\end{eqnarray*}

Suppose $(m_1,\dots,m_5)=(0,1,0,0,1)$ be an input message. For encrypts message $m$, we compute $C_1=s_2\times s_5=104^2 =10816$ and $C_2=u_2\times u_5=1958\times 190 =372020$. For decrypt ciphertext $C=(10816, 372020)$, receiver computes $d=372020\times(10816^{1500}\mod 2579)^{-1} \mod 2579=372020\times 2483\mod 2579=72$. Based on the super-increasing sequence $(2,3,6, 12, 24)$, we have: $72=\underbrace{3}_{a_2}\times \underbrace{24}_{a_5}=\underbrace{6}_{a_3}\times \underbrace{12}_{a_4}$ and therefor the input message $(0,1,0,0,1)$ has two decryptions: itself and $(0,0,1,1,0)$. Therefor, completeness does not holds for the Akleylek et al.'s encryption scheme.
\end{example}
As we mentioned in Remark \ref{re1}, if $p< \prod_{i=1}^n a_i$, then the decryption algorithm does not works properly. In the previous example, since $p>a_2\times a_5$, we have no problem for decryption of the input message.
\begin{example}\label{ex3.3}
Now, consider input message $m=(0,1,1,1,1)$. For encrypt message $m$, one computes $C_1=s_2\times s_3\times s_4\times s_5=104^4 =116985856$ and $C_2=u_2\times u_3\times u_4\times u_5=47252120300$. \sloppy For decrypt ciphertext $(116985856, 27107795330)$, receiver computes $d=27107795330\times(116985856^{1500}\mod 2579)^{-1} \mod 2579 =47252120300\times1479 \mod 2579=26\ne \underbrace{3}_{a_2}\times \underbrace{6}_{a_3}\times \underbrace{12}_{a_4}\times \underbrace{24}_{a_5}$. It is because $p=2579<a_2\times a_3\times a_4\times a_5=5184$.
\end{example}
Therefore in such cases, we cannot efficiently retrieve the input messages from the corresponding ciphertexts.
\section{Cryptanalysis of the Akleylek et al. Cryptosystem}
In this section, we propose our ciphertext-only attack against Akleylek et al.'s cryptosystem to recover message. We also show since
encryption algorithm of the system is deterministic, therefore cryptosystem is not chosen-ciphertext secure. As we previously mentioned, randomness
is the necessary property for CCA2 security, but it is not sufficient.
\subsection{Ciphertext-only attack}\label{ssec4.1}
In this subsection, we show that the Akleylek et al.'s cryptosystem is vulnerable to ciphertext-only attack. In other words, we can obtain message from challenge ciphertext.\\
Suppose $C=(C_1, C_2)$ be any challenge ciphertext which encrypted with this cryptosystem and we wish to find the corresponding message. From equation~\ref{e1}, we have
$C=(C_1, C_2)=\prod_{i=1}^{n}(s_i, u_i)^{m_i}=(s_1, u_1)^{m_1}(s_2, u_2)^{m_2}\dots (s_n, u_n)^{m_n} $. We note that the components $s_i=g^k \mod p$ of the public key are constant
respect to $i$ and we have
\begin{equation}\label{e3}
C_1=\prod_{i=1}^{n}s_i^{m_i}={s_i}^{m_1}\times \ldots \times {s_i}^{m_n}=\underbrace{{s_i}\times \ldots \times {s_i}}_{h-times}={s_i}^h,
\end{equation}
where $h=\sum_{i=1}^{n} m_i$ is the Hamming weight (the number of $m_i=1$) of the input message $m=(m_1, \dots, m_n)$. From equation~\ref{e3}, we can compute the Hamming weight $h$
of the message $m=(m_1,\dots, m_n)$, as the values $s_i$ and $C_1$ are known. Thus, we know the \textit{number} of the $m_i$s, where $m_i=1$. From equation~\ref{e1}, we have
\[
C_2=\prod_{i=1}^{n}u_i^{m_i}={u_1}^{m_1}\times \ldots \times {u_n}^{m_n},
\]
and therefore from $C_2$, we know the number of the $u_i$s where product of them equal to $C_2$, but we do not know which of them. For obtaining these $u_i$s, we need to find a $h$-tuple subset of the $(u_1, \dots, u_n)$ from public key $B=((*, u_1), \dots, (*, u_n))$ such that product of them equals to $C_2$. We denote this subset by $S$. One can chooses $h$ elements of $(u_1, \dots, u_n)$ in $n\choose h$ ways. Therefore, we need at most $n\choose h$ operations to find such subset. After obtaining these $u_i$s, we can obtain original message from the following equation
\[
m_i=\left\{ \begin{array}{ll}
 1 \quad {\rm if}\quad u_i\in S\\
 0\quad {\rm if}\quad  u_i\notin S
  \end{array} \right., \quad 1\leq i\leq n.
\]

\textsc{PROBABILITY OF SUCCESS:} For small $n$, we can efficiently compute ${n\choose h}$. For sufficiently large fixed integer $n$, we provide an upper bound for ${n\choose h}$.
\begin{lemma} Suppose that $h=\lambda n$ is an integer in the range $[0, n]$. Then
\[
{n\choose \lambda n}\leq 2^{nH(\lambda)},
\]
where $H(\lambda)=-\lambda\lg\lambda-(1-\lambda)\lg(1-\lambda)$ is the binary entropy function and $\lg$ is the binary logarithm.
\end{lemma}
\begin{proof} The statement is trivial if $\lambda = 0$ or $\lambda = 1$, so assume that $0 < \lambda < 1$. To prove the upper bound, by the binomial theorem we have
\[
{n\choose \lambda n}\lambda^{\lambda n}(1-\lambda)^{(1-\lambda)n}\leq \sum_{k=0}^n {n\choose k}\lambda^{k}(1-\lambda)^{(n-k)}\leq (\lambda +(1-\lambda))^n=1.
\]
Hence,
\[
{n\choose \lambda n} \leq \lambda^{-\lambda n}(1-\lambda)^{-(1-\lambda)n}=2^{-\lambda n \lg\lambda}2^{-(1-\lambda)n\lg (1-\lambda)}=2^{nH(\lambda)}.
\]
\end{proof}
We show that the number of binary strings of length $n$ with Hamming weight $h=\lambda n$ is bounded by $2^{nH(h/n)}$. Thus, the running time of the proposed attack is ${\mathcal O}(2^{nH(h/n)})$, and depends on the value of $h$. For small and large $h$ i.e., for small and large $\lambda$, $H(\lambda)$ is small and we can efficiently compute ${n\choose h}$ for all $n$. Therefore, if the Hamming weight of the input message is either small or large, we can efficiently break the cryptosystem for all value of $n$. $H(\cdot)$ takes the maximum its value on $\lambda=1/2$, where $H(1/2)=1$. Thus, ${n\choose h}$ takes the maximum its value if $h=n/2$ and the running time of the attack is ${\mathcal O}(2^n)$. Therefore, if $n$ chosen enough large and the input message has Hamming weight close to $n/2$, then the proposed ciphertext-only attack seem does not works. But on the other hand, as we stated in subsection \ref{ssec3.1}, for large $n,h$, {\it completeness} is not holds for the encryption scheme. From equation \ref{e2}, we have $d=\prod_{i=1}^n a_i^{m_i}$. From Lemma~\ref{l1} and \cite{10,11,12,14}, when the $a_i$s are relatively prime, we can \textit{efficiently} calculate $m_i$s from $d$. In the Akleylek et al.'s cryptosystem, since the $a_i$s are super-increasing sequence and are not relatively prime, so small $a_i$s are the divisors of the larger $a_i$s. Thus, as we showed in the example \ref{ex3.2}, we cannot uniquely obtain $m_1, \dots, m_n$ from equation~\ref{e2}. Namely, the problem remains NP-complete and we cannot solve it, especially when $h, n$ is large, i.e., $d$ is the product of the large subset of the $(a_1, \dots, a_n)$.

As a result, for enough large $n$ we have three cases:
\begin{itemize}
  \item[(a)] Input messages with {\it small} hamming weight. In theses cases, we can efficiently compute ${n\choose h}$ and therefore we can apply proposed ciphertext-only attack in polynomial time.
  \item[(b)] Input messages with {\it medium} hamming weight, i.e., $h$ is close to $n/2$. In theses cases, ${n\choose h}$ takes the maximum its value and if $n$ chosen enough large, we cannot efficiently compute it. In such cases, the system has ambiguity and completeness does not holds. Therefore, encryption scheme is not usable.
  \item[(c)] Input messages with {\it large} hamming weight. In theses cases, we can efficiently compute ${n\choose h}$, however, such as previous case, completeness does not holds.
\end{itemize}

\subsection{Chosen ciphertext security}
As we previously stated in the introduction section, the Akleylek et al.'s PKE scheme is deterministic and therefore does not satisfies CCA2 security conditions. Following definition \ref{de2d10}, in the CCA2 security experiment, the challenger runs the key generation algorithm and gives the public key $pk$ to the adversary. The adversary chooses two messages $m_0, m_1$ with $|m_0|= |m_1|$ and gives it to the challenger. The challenger chooses $b \in \{ 0,1\}$ at random and encrypts $m_b$, obtaining the challenge ciphertext $C^*  = {\sf Enc}_{pk}(m_b)$ and gives it to the adversary. Since the encryption algorithm is deterministic, thus each message has one preimage. Therefore, CCA2 adversary simply can compute encryption of $m_0$ with public key $pk$, namely $C={\sf Enc}_{pk}(m_0)$, and then compare it with the challenge ciphertext. If they are equal then $m_b=m_0$, otherwise $m_b=m_1$.

We summarize the results in the following table.
\begin{center}
{\small {\bf Table 2.} Security and Efficiency Analysis of the Akleylek et al.'s Cryptosystem}
\end{center}
\begin{tabular}{|c|c|c|c|}
  \hline
  \quad \textbf{Input Message}\ \ \ \   &\textbf{Proposed Attack} & \textbf{Efficiency}& \textbf{Security} \\
  \hline
  \textsf{\small With small hamming weight}& \textsf{\small Ciphertext-only attack}& ---& \textsf{\small Not secure}\\
  \hline
  \textsf{\small With medium hamming weight}& \textsf{\small Ciphertext-only attack}& $?\,^{1}$& $\approx{\mathcal O}(2^n)$ \\
  \hline
  \textsf{\small With large hamming weight}& \textsf{\small Ciphertext-only attack}& ---& \textsf{\small Not secure}\\
  \hline
  \textsf{\small Any input message}& \textsf{\small CCA attack}& ---& \textsf{\small Not secure}\\
  \hline
\end{tabular}

${1}$. Completeness does not holds.
\section{Modified Cryptosystem}
In this section, we propose our modified encryption scheme based on the Akleylek et al.'s construction.

{$\bullet$ \textbf{Key generation.}  On security parameter $n$, key generator algorithm ${\sf Gen}(1^n)$:
\begin{itemize}
    \item[(a)] Randomly chooses $n$ primes $p_i$ and \textit{safe prime} $p=2q+1$ such that $p > \prod_{i=1}^{n}p_i$. It is clear that $|p|\gg n$.
    \item[(b)] Randomly chooses integers $x, k$ such that $1 < x, k< p-2$ and $\gcd(k, p-1)=1$. Computes
    \end{itemize}
    \begin{eqnarray*}
    y&=&g^x \mod p,\\
    s_i&=&g^k \mod p,\\
    u_i&=&y^k.p_i \mod p,
    \end{eqnarray*}
  and $b_i=(s_i, u_i)$ for $1\leq i\leq n$. Outputs $(n, p, (b_1, \dots, b_n))$ as the public key and $(y, g, x, k,(p_1, \dots, p_n))$ as the private key.
\begin{remark}\label{re5.1}
Note that since $\gcd(k, p-1)=1$, from fact \ref{fa2.6}, $s_i=g^k \mod p$ also is a generator.
\end{remark}
$\bullet$ \textbf{Encryption.} On inputs $m\in \mathbb{Z}_p^*,\, pk$, encryption algorithm {\sf Enc}:
  \begin{enumerate}
  \item[(a)] Uniformly chooses $n$-bit integer $r=(r_1, \dots, r_n)\in \{0,1\}^n$ with $r\ne 0,1$ at random and computes $h=\sum_{i=1}^n r_i$.
  \item[(b)] If $r$ is even then $r'\leftarrow r+1$, else $r'\leftarrow r$.
  \item[(c)] Computes
  \begin{equation}\label{e4}
  C_1=(C_1', C_1'')=\prod_{i=1}^{n} (s_i, u_i)^{r_i} \mod p \quad \textrm{and} \quad C_2=(m+h)^{r'} \mod p,
  \end{equation}
  and outputs $(C_1, C_2)$.
  \end{enumerate}
  It is obviously clear that the modified scheme is chosen-plaintext secure. Each message has $2^n$ corresponding ciphertext, and therefore, the probability of \textit{distinguish}
  between two message is $2^{-n}$ which is negligible.

$\bullet$ \textbf{Decryption.} In the decryption phase, firstly we recover randomness $r'$ was used for encrypts message $m$ from $C_1$. Then $r'$ used to recover message $m$ from $C_2$. It is clear that for correctly recover message $m$, we must recover exact the same randomness $r$ from $C_1$. To recover message $m$ from $(C_1, C_2)$, decryption algorithm {\sf Dec} performs as follows:
  \begin{enumerate}
    \item[(a)] Computes
  \[
   \hat{d}=C_1''.(C_1'^x)^{-1} \mod p= \frac{\prod_{i=1}^{n}u_i^{\hat{r}_i}}{\prod_{i=1}^{n}(s_i^x)^{\hat{r}_i}}\mod p=\prod_{i=1}^n p_i^{\hat{r}_i} \mod p.
  \]
   \item[(b)] Since $p > \prod_{i=1}^n p_i$ and $\hat{r}_i\in \{0, 1\}$, hence $\prod_{i=1}^{n}p_i^{\hat{r}_i}\mod p=\prod_{i=1}^{n} p_i^{\hat{r}_i}$ and so we have
  \[\hat{d}=\prod_{i=1}^{n} p_i^{\hat{r}_i}.
  \]
    Since $\hat{r}_i\in \{0, 1\}$, then $\hat{d}$ is the product of some  distinct  primes $p_i$. By Lemma~\ref{l1}, we conclude that
  \[
        \hat{r}_i=\left\{ \begin{array}{ll}
         1 \quad \textrm{if} \ \ \ p_i\mid d\\
         0 \quad \textrm{if} \ \ \ p_i\nmid d
        \end{array}\right., \quad 1\leq i\leq n.
  \]
  \item[(c)] With retrieved randomness $\hat{r}=(\hat{r}_1, \dots, \hat{r}_n)$ and secret key $(y,k, (p_1, \dots, p_n))$, computes $\hat{h}=\sum_{i=1}^n \hat{r}_i$ and checks wether
  \begin{equation}\label{eq5.2}
   C_1''\stackrel{?}{=}y^{k\hat{h}}\prod_{i=1}^n p_i^{\hat{r}_i} \mod p
  \end{equation}
  holds (\textit{consistency} checking) and rejects the ciphertext if not. If it holds then $r\leftarrow \hat{r}$ and $h\leftarrow \hat{h}$. Note that $C_1''=\prod_{i=1}^n u_i^{r_i} \mod p=y^{kh}\prod_{i=1}^n p_i^{r_i} \mod p$.
  \item[(d)] If $r$ is even then $r'\leftarrow r+1$, else $r'\leftarrow r$.
  \item[(e)] Finds integer $w, 1\leq w\leq p-2$ such that $w\cdot r'=1 \mod p-1$.\,\footnote{Since $|r'|=n\leq|p|$, thus $r'<p $. $r'$ is odd and $p-1=2q$ is even and has two divisor $(2,q)$, therefore, $\gcd(r', p-1)=1$ and $r'$ has multiplicative inverse modulo $p-1$.}
  \item[(f)] Computes $\hat{m}=((C_2)^w \mod p)-h$.
  \item[(g)] Checks wether
  \begin{equation}\label{eq5.3}
  C_2\stackrel{?}{=}(\hat{m}+h)^{r'} \mod p
  \end{equation}
   holds (consistency checking) and rejects the ciphertext if not. If it holds then outputs $m=\hat{m}$.
  \end{enumerate}

\subsection{\bf Security analysis}
\subsubsection{\textsc{Provable Security}}
The basic idea of provable security theory \cite{9} is to reduce the security of a PKE scheme under some attack model to a mathematically hard problem i.e., integer factorization,
discrete logarithm problems and NP-complete problem such as knapsack problem. Provable security has been widely accepted as a standard method for analyzing the security of cryptosystems. Such as original Akleylek et al.'s scheme and previous knapsack-based PKE schemes \cite{5,12,14,15}, we fail to obtain any security proof. In this subsection we nonetheless recall certain security-related facts for the clarity of this paper.
\begin{proposition} If the discrete logarithm problem (DLP) can be computed very efficiently, then the proposed system is not secure.
\end{proposition}
\begin{proof}
First note that even the DLP is computable, we cannot compute $x,k$ from $s_i=g^k \mod p, \,y=g^x \mod p$ and  $u_i=y^k\cdot p_i \mod p$, since $(y, g, x, k,(p_1, \dots, p_n))$ is secret.\\
In the modified cryptosystem, we have
\[
C_1'=\prod_{i=1}^n s_i^{r_i}\mod p=s_i^{\sum_{i=1}^n r_i}\mod p=s_i^h\mod p,
\]
where $h=\sum_{i=1}^n r_i$ and $s_i=g^k \mod p$ is a generator of $\mathbb{Z}_p^*$. If the DLP is computable, then we can determine Hamming weight $h$ from $C_1'=S_i^h \mod p$. According to the discussion in subsection  \ref{ssec4.1}, then the modified scheme is vulnerable to ciphertext-only attack if $n,h$ are small or $h$ is large. In such cases, the adversary can retrieve randomness $r$ from $C_1=(C_1', C_1'')$ and then recover $m$ from $C_2=(m+h)^{r'} \mod p$. Even if the DLP is computable, then the proposed scheme is not completely breaks. The ciphertext-only attack will works for small $(n,h)$ and large $h$. It cannot not break system for large $n$ with medium Hamming weight.
\end{proof}
\begin{proposition} If a certain special knapsack-type problem can be solved
very efficiently, then the proposed system is not secure.
\end{proposition}
\begin{proof} Given $p, u_1, \dots, u_n$ and a ciphertext $C_1=(C_1',C_1'')$, we want to find a
subset $T$ of $\{1, \dots,n\}$ such that
\begin{equation}\label{eq5.4}
\prod_{i\in T} u_i \mod p= C_1''.
\end{equation}
Then we can immediately recover randomness $r$ from $C_1''$ and then compute message $m$ from $C_2=(m+h)^{r'} \mod p$. Finding such a subset $T$ is a kind of knapsack problem.
\end{proof}
Note congruence \ref{eq5.4} is a disguised version of the easy knapsack-type problem of
finding a subset $T$ of $\{1,\dots, n\}$ such that
\[
\prod_{i\in T} p_i \mod p =C_1''.(C_1'^x)^{-1} \mod p,
\]
which we solve by computing $\gcd((C_1''.(C_1'^x)^{-1} \mod p ),p_i)$ for $i = 1, 2,\dots$ .

\textsc{Birthday Attack.} If prime $p$ is chosen too small, then from inequality $p > \prod_{i=0}^n p_i$, it follows that  $n$ is small. Hence $p$ must be sufficiently large to prevent birthday-search through two lists $A$ and $B$ of $2^{n/2}$ elements to find a couple of sets such that:
\[
 \prod_{i\in A} u_i=(\prod_{i\in B} u_i)^{-1}.C_1'' \mod p.
 \]
Therefore $n$ must be chosen such that the adversary's running time is significantly smaller than $2^{n/2}$ steps.
\subsubsection{\textsc{CCA2 Security}}
In this subsection, we show that the modified scheme satisfies CCA2 security. As we showed in subsection \ref{ssec2.4}, the ElGamal system is not CCA2-secure. It is because values $g, y$ are public. Unlike the ElGamal system, in the modified system values $(y, g, x, k,(p_1, \dots, p_n))$ are secret and we cannot perform any modification to the $(C_1',C_1'')$ in order to retrieve randomness $r$. Even if we can perform any modifications to the challenge ciphertext, then the maliciously-formed ciphertexts will be rejected in the scheme's consistency checking step in (\ref{eq5.2}). If we can retrieve randomness $r$, then we can simply recover message $m$.
\begin{theorem}\label{tm5.4}
If the mixed ElGamal-Knapsack encryption scheme is secure, then the modified PKE scheme satisfies CCA2 security in the standard model.
\end{theorem}
\begin{proof} In the proof of security, we exploit the fact that for a well-formed ciphertext,‎ we can recover the message if we know the \textit{randomness} $r$ that was used to create
the ciphertext.‎\\
In the  CCA2 experiment (Definition \ref{de2d10}), the challenger runs the key generation algorithm and gives the public key $pk$ to the adversary.\\
\textsf{Challenge Ciphertext}. The adversary chooses two messages $m_0, m_1$ with $|m_0|= |m_1|$ and gives it to the challenger. The challenger chooses $b \in \{ 0,1\}$ at random,
randomness $r^*$ and encrypts $m_b$, obtaining the challenge ciphertext $C^*=(C_1^*,C_2^* )$, where $C_1^*=\prod_{i=1}^{n} (s_i, u_i)^{r_i^*} \mod p$
and $C_2^*= (m_b+h^*)^{r'^*} \mod p$ and gives it to the adversary, where $h^*$ is the Hamming weight of the randomness $r^*$.
We denote by $r^{*}$ the corresponding intermediate quantity chosen by the challenger.

The challenger has to simulate the decryption oracle. The CCA2 adversary submits a request $C= (C_1, C_2)$ to the challenger, and it outputs decryption of the queried ciphertext to the adversary. He attempts to guess
the challenge bit $b$ based on the output of the challenger. In the CCA2 experiment, the adversary is not allowed to ask the decryption of
the challenge ciphertext, but can obtain the decryption of any modified ones based on the challenge
ciphertext.\\
To investigate CCA security experiment, we consider two potential cases chosen by the adversary for querying from the challenger. We also show
that any modification to the challenge ciphertext does not reveal any useful information about the challenge message $m_b$.\\
\textsf{Case 1:} $C_1=C_1^*$ and $C_2\ne C_2^*$. In this case, the adversary chooses $C_2$ at
random and queries on ciphertext $(C_1^*,C_2)$. The challenger takes as input $(C_1^*,C_2)$ and computes $r={\sf Dec}_{pk}(C_1^*)=r^*$, $h=h^*$ and
$r'=r'^*$. It also computes $\hat{m}=((C_2)^{{\sf Inv}(r'^*)} \mod p)-h^* \ne ((C_2^*)^{{\sf Inv}(r'^*)}\mod p)-h^*=m_b$,
where ${\sf Inv}(r')=(r')^{-1} \mod p-1$ is the multiplicative inverse of $r$. Since $C_2^*=(m_b+h^*)^{r'^*} \mod p\ne (\hat{m}+h^*)^{r'^*} \mod p$, thus the simulator rejects
the ciphertext in (\ref{eq5.3}).  Therefore, the system does not reveal any information about the challenge
message $m_b$, and so, advantage of the adversary to guess the challenge bit $b$ in this case is zero.\\
In this case, the adversary cannot perform any modification to $C_2$ based on $C_2^*$ in order to retrieve $m_b$,
since he does not know the internal random component $r^*$ was chosen by the challenger for encrypts $m_b$.

\textsf{Case 2:} $C_1\ne C_1^*$ and $C_2= C_2^*$. In this case, the adversary chooses $C_1$ at random and queries on ciphertext
$(C_1,C_2^*)$. The challenger takes as input $(C_1,C_2^*)$ and computes $r={\sf Dec}_{pk}(C_1)$. Since encryption algorithm of
$C_1$ is deterministic, therefore any randomness $r$ has \textit{one} preimage.
Thus if $C_1\ne C_1^*$, then $r={\sf Dec}_{pk}(C_1)\ne {\sf Dec}_{pk}(C_1^*)=r^*$. In the worst case, we assume $r$ and $r^*$
have the same Hamming weight, namely $h=h^*$. So, we have $\hat{m}=((C_2^*)^{{\sf Inv}(r')} \mod p)-h^*\ne ((C_2^*)^{{\sf Inv}(r'^*)}\mod p)-h^*=m_b$.
Hence, the simulator rejects the ciphertext in (\ref{eq5.3}), since $C_2^*=(m_b+h^*)^{r'^*} \mod p\ne (\hat{m}+h^*)^{r'} \mod p$. Therefore, the encryption scheme does not reveal any information about the challenge
message $m_b$ and so, advantage of the adversary to guess the challenge bit $b$ in this case is zero.\\
As shown in \cite{7,18}, in the knapsack-based PKE schemes, CCA2 adversary cannot efficiently produces legitimate
ciphertext based on $C_1^*$. As shown in \cite{18}, the probability of succeed adversary for retrieve $r$ with {\it one}
bit differ from $r'^*$ is $1/2n$ which is smaller than 1/2 (note in general, the probability of guessing $b$ is 1/2; $b=0$ or $b=1$).
We stress that even if the adversary can computes $r$ with probability greater than $1/2$, then since the retrieved
randomness $r$ is not equal to $r^*$ (differ from one bit), therefore $\hat{m}=((C_2^*)^{{\sf Inv}(r')}\mod p)-h^*$ is not
equal to $m_b$, where we assume $r$ and $r^*$ have the same Hamming weight. So, as we state above, the simulator will rejects the ciphertext in (\ref{eq5.3}).
\end{proof}
\section{Conclusion}
In this paper, we consider a knapsack-based PKE scheme mixed with the ElGamal cryptosystem. This cryptosystem uses the ElGamal system in the key generation stage to disguise the secure knapsack (super-increasing sequence) in order to produce the public knapsack. It uses subset product (multiplicative knapsack) problem as encryption function which is NP-complete
problem. We showed that this combination leaks the security and makes the cryptosystem vulnerable to ciphertext-only attack. In addition, since encryption algorithm for the mentioned scheme
is deterministic, therefore it does not satisfy CCA2 security requirements. Thus, the resulting encryption scheme is also vulnerable to man-in-the-middle attack, and therefore, the
scheme is not suitable to implement in a P2P network.
Besides, as we showed, completeness property does not holds for the system in the general.\\
We modified this cryptosystem to improve its security and efficiency. The modified scheme is CCA2-secure and the proposed ciphertext-only attack is not applicable. Completeness holds for all cases and each ciphertext decrypts uniquely.

\textbf{Acknowledgment}

We thank the anonymous referees for insightful comments.

\end{document}